\documentclass[conference]{IEEEtran}

\IEEEoverridecommandlockouts
\usepackage{cite}
\usepackage{amsmath,amssymb,amsfonts}
\usepackage{graphicx}
\usepackage{textcomp}

\usepackage{balance}
\usepackage{algorithm}
\usepackage{algpseudocode}

\usepackage{url}

\usepackage{amsthm}
\newtheorem{theorem}{Theorem}
\newtheorem{lemma}{Lemma}

\newtheorem{definition}{Definition}

\usepackage{mathtools}
\usepackage[table]{xcolor}    
\usepackage[para]{threeparttable}

\usepackage{times}  
\newcommand{\BAStar}{BA{\fontfamily{lmr}\selectfont$\star$}}
\def\BibTeX{{\rm B\kern-.05em{\sc i\kern-.025em b}\kern-.08em
    T\kern-.1667em\lower.7ex\hbox{E}\kern-.125emX}}

\usepackage{color}
\definecolor{mycolor}{rgb}{0.5,0,1}

\usepackage{caption}
\usepackage{subcaption}

\begin{document}
\title{On Incentive Compatible Role-based \\Reward Distribution in Algorand \vspace{-0.25cm}
}

\author{Mehdi Fooladgar$^{\dag \star}$, Mohammad Hossein Manshaei$^{\dag \ddag \star}$, Murtuza Jadliwala$^\diamond$, and Mohammad Ashiqur Rahman$^\ddag$\\
$^\dag$Department of Electrical and Computer Engineering, Isfahan University of Technology, Iran\\
$^\diamond$Department of Computer Science, University of Texas at San Antonio, USA\\
$^\ddag$Department of Electrical and Computer Engineering, Florida International University, USA\\
Emails:  m.fooladgar@ec.iut.ac.ir, manshaei@iut.ac.ir, murtuza.jadliwala@utsa.edu, \{marahman,mmanshae\}@fiu.edu.\\
\thanks{$^\star$M. Fooladgar and M. H. Manshaei are equally contributing authors.}}

\maketitle
\begin{abstract}
Algorand is a recent, open-source public or permissionless blockchain system that employs a novel proof-of-stake byzantine consensus protocol to efficiently scale the distributed transaction agreement problem to billions of users. In addition to being more democratic and energy-efficient, compared to popular protocols such as Bitcoin, Algorand also touts a much high transaction throughput. Despite its promise, one relatively under-studied aspect of this protocol has been the incentive compatibility of its reward sharing approach, without which cooperation cannot be guaranteed and the protocol will fail in a practical environment comprising of rational users. This paper is the first attempt in the literature to study and address this problem. By carefully modeling the participation costs and rewards received within a strategic interaction (or game) scenario, we first empirically show (by means of simulations) that even a small number of nodes defecting to participate in the protocol tasks due to insufficiency of the available incentives can result in the Algorand network failing to compute and add new blocks of transactions. We further show that this effect, which was observed in simulation experiments, can be formalized by means of a mathematical (game-theoretic) model of interaction in Algorand given its participation costs and the current (or planned) reward distribution/sharing approach envisioned by the Algorand Foundation. Specifically, on analyzing this game model we observed that mutual cooperation under the currently proposed reward sharing approach is not a Nash equilibrium. 
This is a significant result which could threaten the success of an otherwise robust distributed consensus mechanism. To remedy this problem, we propose a novel reward sharing approach for Algorand and formally show that it is incentive-compatible, i.e., it can guarantee cooperation within a group of selfish Algorand users. Extensive numerical and Algorand simulation results further confirm our analytical findings. Moreover, these results show that for a given distribution of stakes in the network, our reward sharing approach can guarantee cooperation with a significantly smaller reward per round. Our protocol also helps designers to better react to different conditions in the network, where the distribution of stakes can potentially converge to a high population of nodes with small stakes. 
\end{abstract}


\section{Introduction}
\label{sec:Intro}

%

A \emph{blockchain} is an immutable distributed database that records a time-sequenced history of facts called transactions. This record is maintained by constructing and maintaining consistent copies of the cryptographic hash-chain of transaction blocks (or sets) in a distributed fashion. One key aspect of any blockchain protocol is the \emph{consensus} algorithm that enables agreement among a distributed network of autonomous \emph{nodes} (a.k.a. \emph{miners} in certain protocols) on the state of the blockchain, under the assumption that a fraction of them could be malicious or faulty. Blockchains could be further categorized as \emph{permissioned} or \emph{permissionless} depending on whether a trusted infrastructure exists or not to establish verifiable identities for network nodes. 

In \emph{Bitcoin} \cite{nakamoto2008bitcoin}, a popular permissionless blockchain protocol, consensus is achieved by the network selecting a \emph{leader node} or \emph{block proposer} in an unbiased fashion once every 10 minutes on an average (named as a \emph{round}). The selected leader node gets the right to commit or append a new block onto the blockchain. The network then implicitly accepts this block by adding the next block on top of it or reject it by choosing some other block. Consensus in Bitcoin is thus long-term, i.e., a block is said to be ``included" in the blockchain if it has received a significant number of confirmations\footnote{Number of blocks added on top of the block in question in the longest valid blockchain.}. The Bitcoin protocol uses a \emph{proof-of-work} (PoW) mechanism to select the leader in each round, in which each node or miner attempts to solve a hash puzzle and who solves it first is selected and gets the right to propose the next block.
As PoW involves significant computation, the Bitcoin protocol includes a reward mechanism to incentivize miners to compete in a fair manner and to behave honestly. Besides Bitcoin, several other distributed systems (e.g., \emph{Ethereum} \cite{ethereum} and other alt-coins \cite{allcrypto}) also employ a Bitcoin-like PoW-based consensus algorithm and a reward model to ensure honest participation.


The wide-scale growth and adoption of Bitcoin, both in terms of users and miners, exposed several significant and inter-related issues with its PoW-based consensus mechanism. In particular, the hash puzzle-based PoW approach is wasteful in term of energy perspective \cite{DwyerM:2014}, it does not prevent forking in the blockchain, results in mining (and hash power) centralization \cite{Rosenfeld:2011}. More importantly, it does not scale well with the number of transactions and network users \cite{Croman:2016}. For instance, the transaction rate of Bitcoin is currently only 7 transactions per second, which is significantly lower than the transaction rates afforded by centralized transaction processing systems such as PayPal (450 transactions per second) and VisaNet (between 1667 and 56,000 transactions per second) \cite{bitcoinscalability}. It is clear that current Bitcoin transaction throughput is not sufficient for many practical applications. Several platform specific efforts, such as BIP 102~\cite{bip102} and Bitcoin-NG~\cite{eyal2016bitcoin}, have been proposed to improve Bitcoin’s transaction throughput. 
Alternatively, other platform-agnostic solutions aimed to improve the scalability-related shortcomings of PoW-based consensus by employing a committee or sharding approach \cite{luu2016secure}, payment networks \cite{lightingnetwork,raidennetwork}, and side-chains \cite{back2014enabling}. Other approaches have tried to either improve the existing version of PoW itself \cite{primecoin:2013} or have proposed alternatives such as Proof-of-State (PoS) \cite{ppcoin:2012,blackcoin:2014,reddcoin:2014,post,chen2016algorand}, Proof-of-Burn (PoB) \cite{slimcoin:2014}, Proof-of-Elapsed Time (PoET) \cite{hyperledger} and Proof-of-Personhood (PoP) \cite{BorgeKJGGF:2017}. Some other efforts \cite{SompolinskyLZ:2016,tangle:2018} have attempted to improve scalability and throughput by implementing the distributed block ledger as a Directed Acyclic Graph (DAG), rather than a linear hash chain as done by popular systems such as Bitcoin.

%
Of all the above efforts, the Algorand~\cite{chen2016algorand,gilad2017algorand}  has garnered the most attention within the permissionless blockchain and cryptocurrency community, primarily because of its innovative PoS-based consensus or byzantine agreement protocol that is not only computationally (and energy) efficient, but also provides strong security guarantees against forking in a network comprising of faulty and malicious users or nodes\footnote{The term users and nodes are used interchangeably. Typically, \emph{users} control \emph{nodes} which are computational systems that are part of the Algorand peer-to-peer network and execute the reference software.}. Algorand eliminates the possibility of hash power centralization by removing the difference between normal network users and miners and scales pretty well. In fact, Algorand can commit about 750 MBytes of transactions per hour, which is 125 times of Bitcoin's throughput~\cite{gilad2017algorand}.
These security and performance guarantees of the Algorand consensus design have resulted in a lot of optimism within the blockchain community. However, one critical issue has not received much, if any, attention: \emph{does the currently proposed Algorand reward distribution approach promote participation or cooperation, and not defection, among rational users to complete all the required protocol tasks? In other words, is the current Algorand reward distribution approach incentive-compatible?}


Since the inception of Bitcoin, a significant effort has been spent by the research community towards understanding the incentive-compatibility of the Bitcoin's reward distribution approach \cite{biais2019blockchain, kiayias2016blockchain,kwon2017selfish}, towards characterizing the strategic behavior of rational miners in mining pools \cite{courtois2014subversive,kiayias2016blockchain,luu2015power,eyal2015miner,lewenberg2015bitcoin,schrijvers2016incentive}, and towards designing new incentive-compatible PoW-based cryptocurrencies \cite{DwyerM:2014} and scalability solutions \cite{manshaei:2018,kim2019two}. Such a thorough analysis under the additional assumption of rationality has helped to significantly mature the permissionless blockchain technology. However, no such analysis for the  Algorand exists yet, and this paper attempts to fill this research gap.

Here, we make the first attempt to formally analyze the Algorand's reward distribution strategy by employing well-established game-theoretic tools and techniques. More specifically, by modeling a single round of the Algorand's consensus or byzantine agreement protocol as a single stage non-cooperative multi-player game, we show that without an efficient reward sharing protocol, nodes are willing to deviate from cooperation and behave selfishly. Motivated by the need for solving this deviation problem, we propose a new reward distribution approach for Algorand, which is in addition to the stake possessed by the users, the approach considers their role in the byzantine agreement protocol in order to distribute the per-round rewards. We further show that our proposed role-based reward distribution approach is able to converge to a Nash equilibrium (NE) where a certain subset of nodes will cooperate. We  design a reward sharing mechanism based on our results and implement it in an Algorand simulator. We conduct extensive empirical evaluation of the proposed reward distribution approach using both a numerical and Algorand protocol simulations. Our empirical evaluation further confirms our analytical results by showing that we can distribute significantly smaller rewards among users and also enforce cooperation in Algorand. The Algorand Foundation can use our results to keep track of the network state and adapt reward accordingly. They can also employ our analytical tool to adapt the best reward distribution in the  future, when they want to distribute the transaction fees among Algorand users. To the best of our knowledge, this paper is the first to provide a systematic analysis of incentive design in Algorand. 


This paper is organized as follows. In Section~\ref{sec:SysModel}, we discuss the state of the art of Algorand. In Section~\ref{sec:ProblemFormulation}, we present the incentive design problem for Algorand. Section~\ref{sec:GameModel} presents the game model and its analysis. We also propose our novel reward sharing approach in this section, following by our evaluations in Section~\ref{sec:Simulation}. Finally, Section~\ref{sec:Conclusion} summarizes the paper.

\section{Algorand System Model}
\label{sec:SysModel}

\begin{figure*}[ht]
\centerline{\includegraphics[scale=0.7]{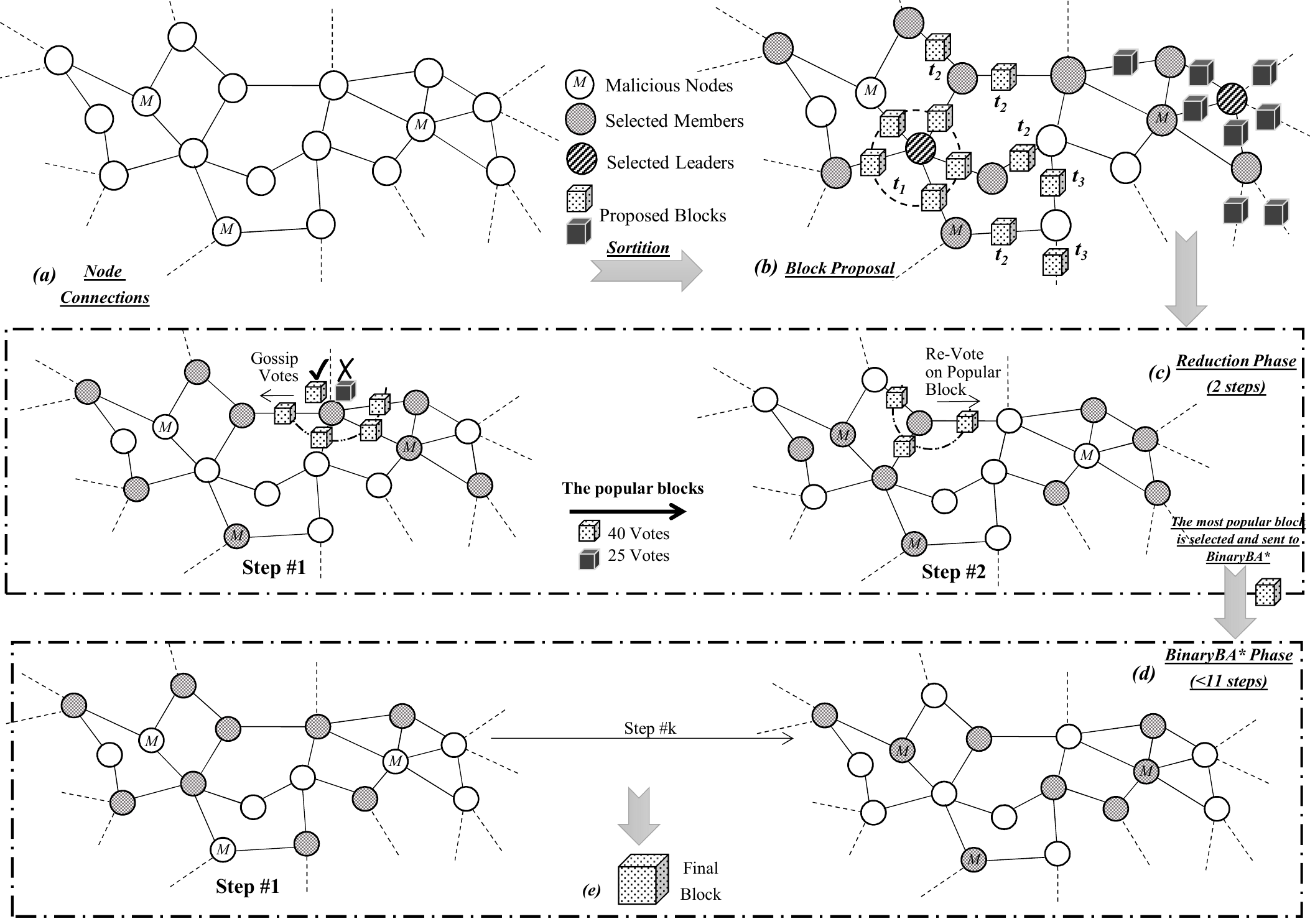}}
\caption{Algorand System Model. \emph {(a)} Algorand nodes build a peer-to-peer network which  includes both malicious and honest nodes. \emph{(b)} In any given time slot, each user executes cryptographic \emph{sortition algorithm} to determine his role in that time slot. At time $t_1$ each leader sends his proposed block to all first hop neighbors. Consequently, all nodes forward their received blocks to their neighbors in the following time slots. \emph{(c)} Reduction phase reduces consensus problem to agreement on one or two options. Committee members select block with the highest priority among their received blocks and gossip their votes for that. In $step \#1$, committee members vote for highest priority block they received. In $step\#2$, new committee members count last step votes and re-transmit popular blocks as their vote in the second step. Output of this phase can be an empty block in case of not-receiving minimum number of required  votes. \emph{(d)} Binary\BAStar~phase reach agreement on a proposed block from reduction phase or an empty block. Note that the figure represents the case where the network has strong synchrony and agrees on a final block. This is usually completed in the first step, i.e., $k=1$. But this phase can be followed up (in average for 11 steps, i.e., $k=11$) to ensure that each node agrees on a same consensus. In each step committee members votes for their observation of the reduction phase. \emph{(e)} the final block would be added to the chain.}
\label{fig:SysModel}
\end{figure*}


In this section, we first summarize the Algorand protocol. This description is intended to provide readers with the main concepts of Algorand as it relates to this paper. Interested readers are referred to \cite{chen2016algorand} for more technical details on Algorand and to \cite{nakamoto2008bitcoin} for Bitcoins and PoW-based blockchains. We begin by first contrasting the consensus approaches employed in PoW (e.g., Bitcoin) and PoS (e.g., Algorand) blockchains. Then, we outline the network (communication) protocol of Algorand followed by details of its consensus mechanism.

\subsection{Contrasting Consensus in PoW with PoS Blockchains}
Due to its open and distributed nature, consensus or agreement (on transaction blocks) in current public blockchain solutions such as Bitcoin have relied on a proof-of-work (PoW) approach, where users must repeatedly compute hashes to grow the blockchain, and the longest chain is considered authoritative. However, there are several significant shortcomings of such an approach:
\begin{itemize}
    \item PoW wastes significant amount of computation, and by relation, the electrical energy used to achieve it. PoW schemes also assume that a majority of the nodes contributing to the network's hash (or computational) power are honest, i.e., at least 51\% of the network's hash power comes from honest users/nodes.
    \item PoW-based consensus eventually leads to Concentration of Power, where entities in the network eventually monopolize computational power to control new block addition (e.g., Bitcoin Mining Pools) \cite{lewenberg2015bitcoin, eyal2015miner, luu2015power}.
    \item PoW allows the possibility of forking where two different hashchains could reach the same length and neither one supersedes the other \cite{kiayias2016blockchain,courtois2014subversive}. Efforts to mitigate the impact of forking in existing solutions has resulted in the block inter-arrival and transaction confirmation times to become impractically high (e.g., current Bitcoin block inter-arrival time is 10 minutes while transaction confirmation time is 1 hour). As a result, current PoW blockchain solutions do not scale well with the number of transactions and users.
\end{itemize}
To overcome these shortcomings, Algorand proposes a novel proof-of-stake (PoS) based consensus protocol.
Similar to Bitcoin, Algorand is fully decentralized and maintains a public, immutable ledger of transactions by reaching consensus on the order of transactions in the ledger. However in Algorand all users are ``equal'', i.e, there is no distinction between users (miners) who can add new blocks and those who just create and receive transactions. Thus, Algorand is more democratic! Moreover, as each user/node runs the same computationally efficient functions to achieve consensus (as opposed to PoW-based systems where users compete for the right to add the next block), Algorand does not waste computations, and thus electricity. Lastly, the design of Algorand's consensus protocol guarantees that there is no forking with an overwhelmingly high probability. A side-effect of this is that Algorand scales extremely well with the number of users/nodes and transactions, compared to classical PoW-based systems \cite{gilad2017algorand}.

\subsection{Summary of Algorand}
Next, we summarize the creation, distribution and agreement of transaction blocks in Algorand, as shown in Figure~\ref{fig:SysModel}.

\subsubsection{Assumed Adversary Model}
\label{subsub:Adversary}
In addition to standard cryptographic assumptions, Algorand assumes that \emph{honest users} always run bug-free reference software and consequently follow all defined steps by Algorand. As is standard in PoS systems, Algorand assumes that the fraction of money held by honest users is above some threshold $h$ (a constant greater than $\frac{2}{3}$). An adversary can participate in Algorand by creating multiple sybil nodes/users and owning some money or stake in the system. An adversary in Algorand can arbitrarily corrupt honest users, provided that the amount of money held by honest, non-compromised users remain above the threshold $h$. However an adversary cannot compromise the keys of honest non-compromised users. Algorand assumes that most honest users receive messages sent by most other honest users within a known time bound in order to continue to make progress on adding blocks to the blockchain (i.e., \emph{liveniness goal}). This is the \emph{strong synchrony} assumption. Algorand can achieve consensus or agreement on blocks (i.e., \emph{safety goal}) even if the network is \emph{asynchronous} (or controlled by the adversary) for a long but bounded period of time, provided it is strongly synchronous for a long period of time after that. 

\subsubsection{Network and Communication Protocol}
The Algorand network is a \emph{peer-to-peer} network of honest and faulty or malicious nodes, where each node is represented by a public/private key pair (see Figure~\ref{fig:SysModel}\emph{-(a)}). The number of malicious or faulty nodes is bounded by the honest stake ownership condition outlined earlier. Nodes in the network communicate in a peer-to-peer fashion using unique TCP connections. Communications happen by means of a standard \emph{gossip protocol} where each node broadcasts his message to all his peers, who in turn relay it to their neighbors. The Algorand communication protocol defines four types of messages: 
 \begin{itemize}
        \item \emph{Transaction}: this message transfers a certain amount of \emph{Algos} (currency unit in the Algorand system) from a sender to a receiver (identified with their public-keys) and signed by the sender (with its private key), which is referred to as a transaction. Multiple transactions are organized into a \emph{block}. An Algorand block is either a set of transactions or an empty block. In addition, each block contains a pre-determined random seed (described later) and the hash of the previous consensus or agreed block it is extending. 
        \item \emph{Voting}: this message contains a signed vote by the sender along with the \emph{sortition proof} (described below). Each sortition proof is associated with a priority value which is computed in a deterministic fashion.
        \item \emph{Block proposal}: this message contains a new Algorand block (to be added), along with the signed hash of the block and a sortition proof establishing the role of the sender as a block proposer or \emph{leader}.
        \item \emph{Credential}: this message contains the sortition proof of the block proposer or leader, which is generally broadcast at the beginning of each round by the leader using the gossip protocol. Peer nodes employ the priority values extracted from sortition proofs in the credential messages to avoid relaying block proposals with low priorities. This helps preventing congestion in the network due to a significantly large number of block proposals.
 \end{itemize}

\subsubsection{Consensus or Byzantine Agreement (\BAStar)}
Algorand's consensus or Byzantine Agreement (\BAStar) protocol operates in \emph{rounds}, where in each round all nodes attempt to reach agreement on a new block of transactions. At the beginning of a round, each node employs cryptographic sortition to privately determine if it is a block proposer or leader, i.e., has the right to propose a block for that round. 
To propose a block, each leader node assembles the pending and validated transactions inside a block proposal and gossips it together with its sortition proof of being elected a leader (see Figure~\ref{fig:SysModel}\emph{-(b)}). 
After block proposals are broadcast, each node collects incoming block proposals for a fixed duration, selecting and retaining the one valid block with the highest priority sortition proof. 

    
Each user then (asynchronously) initializes the \BAStar~protocol with the highest-priority block they have received. The \BAStar~protocol enables all nodes in the network to reach consensus on a single block. The \BAStar~protocol comprises of two sequential phases, the \emph{Reduction phase} (Figure~\ref{fig:SysModel}\emph{-(c)}) followed by the \emph{Binary\BAStar~phase} (Figure~\ref{fig:SysModel}\emph{-(d)}), with each phase consisting of a sequence of \emph{steps}. At a high level, in each step first a random or unpredictable group of nodes referred to as \emph{committee members} is elected. 
Then the elected committee members \emph{vote} on a specific block, based on votes received from the previous step, and broadcast their new votes in voting messages. Readers should recall that all voting messages also contain a sortition proof which proves the validity of the broadcaster as a committee member. 
\begin{itemize}
\item \emph{Reduction Phase (Figure~\ref{fig:SysModel}\emph{-(c)})}:  This first phase of the \BAStar~protocol comprises of \emph{exactly two} steps. In the first step, each committee member votes for the hash of the blocks proposed for consideration. In the second step, committee members vote for the block hash that received votes over a certain threshold. 
If no block hash receives enough votes, committee members vote for the hash of the default empty block. Reduction phase concludes with either at most one non-empty block hash (the one that received the maximum number of votes above the threshold) or hash of an empty block (if no block hashes received enough votes). This output of the reduction phase is passed as input to the Binary\BAStar~phase.
\item \emph{Binary\BAStar~Phase (Figure~\ref{fig:SysModel}\emph{-(d)})}: 
The goal of the Binary\BAStar~Phase is to reach agreement or consensus on the majority voted non-empty block (hash) from the reduction phase or, in case there is no consensus, on the empty (default) block. In the common case, when the network is strongly synchronous and the block proposer or leader was honest, Binary\BAStar~phase will start with the same block hash for most users, and will reach consensus in the first step, since most committee members
vote for the same block hash value. If the network is not strongly synchronous, Binary\BAStar~may return consensus on two different blocks (i.e., the block received from the reduction phase and the empty block).
\end{itemize}
The outcome of the Binary\BAStar~phase is used by the \BAStar~algorithm to arrive at either a \emph{final} or a \emph{tentative} consensus. 
Final consensus means that \BAStar~will not reach consensus on any other block for that round, while tentative consensus means that \BAStar~is unable to guarantee the safety goal in this round, either because of network asynchrony or due to a malicious block proposer or leader.

\subsubsection{Cryptographic Sortition}
Each node in the network employs a cryptographic sortition algorithm to determine if it is selected as a leader (or block proposer) at the beginning of each round, and later, if it is selected as a committee member at the beginning of each step (of both the Reduction and Binary\BAStar~ phases). The sortition algorithm is implemented using \emph{Verifiable Random Functions (VRF)} \cite{micali1999verifiable} which allow users to produce verifiable proofs using their private keys that can be publicly verified using the corresponding public key. Specifically, in order to generate a sortition proof for step $s$ in round $r$, a user $i$ computes $sig_i(r,s,Q^{r-1}),$ where $sig_i$ is a digital signature computed using the user $i$'s private key, and $Q^{r-1}$ is a random seed (predetermined at the end of the previous round, i.e., $r-1$). This sortition proof is included by the nodes in their block proposals and voting messages in order to prove their roles as leader and committee members, respectively. The recipients of these messages first validate the signature (using the public key) and then compute the hash of the sortition proof to verify a certain sortition condition that determines the validity of the claimed role. The possibility that the condition
is verified is directly proportional to the stake of the node (to which the proof belongs to) and depends on a constant role parameter fixed in the protocol. Due to space restrictions, we will not further elaborate on this verification condition and interested readers can find more technical details in~\cite{gilad2017algorand, micali1999verifiable,chen2016algorand}. In summary, cryptographic sortition is a simple, offline and lightweight process that involves computing a single hash and a digital signature in order to verify a node's role, and thus, its eligibility to propose a block and/or vote in the Algorand consensus (\BAStar) protocol.

\section{Incentive Design in Algorand: Problem Formulation and Motivations}
\label{sec:ProblemFormulation}

\begin{table}[t]
	\centering
	\rowcolors{2}{gray!25}{white}
	\caption{List of Symbols in Algorand Analysis}
	\begin{tabular}{c|l}
		\hline
		\textbf{Symbol} & \textbf{Definition} \\ \hline 
		$R_i$ & Foundation reward in round $t_i$\\
		$F_i$ & Summation of transaction fees in round $t_i$\\
		$P_i^{F}$ & Reward pool level in round $t_i$\\
		$B_i$ & The shared rewards in round $t_i$\\
		$\alpha$ & Fraction of rewards shared between leaders \\
		$\beta$ & Fraction of rewards shared between committee members\\
		$\gamma$ & Fraction of rewards shared between remaining online nodes\\ \hline
		$c^{fix}$ & Common costs of Algorand nodes \\
		$c^L$ & Costs for Algorand leaders\\
		$c^M$ & Costs for committee members \\
		$c^K$ & Costs for Algorand remaining online nodes\\ \hline
		$r_i^L$ & Rewards per each unit of stake for a leader \\
		$r_i^M$ & Rewards per each unit of stake for a committee member \\
		$r_i^K$ & Rewards per unit of stake for a remaining online node \\
		$s_{j}$ & Stake of node $j \in \{l_j,m_j,k_j\}$; $s_{l_j}$ is reward for leader $l_j$ \\ \hline
		$S_L$ &  Summation of all stakes for leaders; i.e. $S_L = \sum_{j \in L}{s_{l_j}}$\\
		$S_M$ &  Summation of all stakes for committee members\\
		$S_K$ &  Summation of all stakes for other nodes\\
		$S_N$ &  Summation of all stakes, i.e., $S_N = S_L + S_M + S_K$\\ \hline
		$u_i^{l_j}$ & Payoff for leader $l_j$ in round $t_i$\\
		$u_i^{m_j}$ & Payoff for committee member $m_j$ in round $t_i$\\
		$u_i^{k_j}$ & Payoff for remanding node $k_j$ in round $t_i$\\ \hline
	\end{tabular}
	\label{tab:symbol-table}
\end{table}

Similar to any permissionless blockchain-based cryptocurrency protocol, Algorand must also provide enough incentives to foster cooperation among its participants, whether they are leaders, committee members, or online nodes, in order to enable effective consensus (on the set of transactions). In this section, we first summarize all processing costs defined in the Algorand Byzantine consensus protocol, followed by a discussion of how rewards could be distributed among the various protocol participants. Finally, we empirically show that if rewards are not appropriately distributed in Algorand, rational participants have an incentive to not cooperate (in the consensus protocol tasks), resulting in no new blocks being added. Our goal here is to highlight the need for designing a \emph{incentive-compatible reward sharing mechanism} for achieving cooperation in Algorand. Note that Table~\ref{tab:symbol-table} presents the notations used throughout the paper.

\subsection{Algorand Costs}
\label{sub:Costs}

Based on the summary of the Algorand operation outlined in the previous section, it is clear that each participant or user, irrespective of their role, is expected to perform some processing and communication tasks during each phase of the protocol which incurs some measurable cost, say, in terms of consumed energy. These costs for each processing tasks can be further quantified using monetary values (e.g., \emph{Dollars} or \emph{Algos}) by using the current energy costs. Below we present a brief description of each of these tasks that incur some significant cost, which are also summarized in Table~\ref{tab:my_label}. We would like to stress that our goal here is not to precisely quantify the cost associated with each task - we simply argue that each of these tasks incurs a significant cost which can be easily quantified and is abstracted by us as the corresponding cost parameter.

{\bf Transaction Verification Cost} ($c^{ve}$): This cost is incurred by an Algorand node to check the validity of a transaction. For each transaction validity check, the node must verify the signature and also check whether the sending user has enough \emph{Algos} in its account for a successful transaction. A leader assembles a set of transactions into a block, and all Algorand nodes check the validity of transactions inside a block.

{\bf Seed Generation Cost} ($c^{se}$): Algorand requires a random and publicly known seed as an input to the sortition algorithm. Thus, a new seed is published in each round of Algorand. This seed is a random number generated by VRF \cite{micali1999verifiable} from the last seed value and the current round number. Also, for security concerns Algorand refreshes the seed every $R$ rounds \cite{gilad2017algorand}. We parameterize the cumulative cost of generating a new seed in each round as $c^{se}$.

{\bf Sortition Algorithm Cost} ($c^{so}$): As outlined in the previous section, the sortition algorithm employs a VRF function~\cite{micali1999verifiable} to generate a membership proof which is included by leaders and committee members in their messages to prove their role (leader or committee member). The cost of processing all the tasks of the sortition algorithm is parameterized as $c^{so}$.

\begin{table}[t]
	\centering
	\rowcolors{2}{gray!25}{white}
	\caption{Algorand tasks and costs given the role of nodes.}
	\footnotetext[1]{Every $R$ rounds}
	\begin{threeparttable}
		\begin{tabular}{l|c|c|c|c}
			\hline
			{\bf Task} & {\bf Symbol} & {\bf Leader} & {\bf Committee} & {\bf Others} \\ 
			 &  & {$l_j$} & {$m_j$} & { $k_j$} \\ \hline
			Transaction Verification & $c^{ve}$ & \checkmark & \checkmark & \checkmark \\
			Seed Generation & $c^{se}$ & \checkmark & \checkmark & \checkmark \\
			Sortition Algorithm & $c^{so}$  & \checkmark & \checkmark & \checkmark \\
			Verify Sortition Proof & $c^{vs}$ & \checkmark & \checkmark & \checkmark \\
			Block Proposition & $c^{bl}$  & \checkmark & & \\
			Gossiping & $c^{go}$ & \checkmark & \checkmark & \checkmark \\
			Block Selection & $c^{bs}$  & & \checkmark& \\
			Vote & $c^{vo}$  & & \checkmark &  \\
			Vote Counting & $c^{vc}$ & \checkmark & \checkmark & \checkmark \\
		\end{tabular}
		\begin{tablenotes}
		\end{tablenotes}
	\end{threeparttable}
	\label{tab:my_label}
\end{table}



{\bf Block proposition Cost} ($c^{bl}$):  The cost of assembling a set of outstanding, but valid, transactions (including the sortition proof) into a block and broadcasting it to the neighboring nodes in the network is borne by each (selected) leader node in each round. We parameterize this cost as $c^{bl}$.

   
{\bf Gossiping Cost} ($c^{go}$): During each round, each node in the Algorand network supports the network by forwarding (gossiping) network messages, including transactions, blocks and votes. A cumulative expected cost for each round for each node is parameterized by us as $c^{go}$.


 {\bf Block Selection Cost} ($c^{bs}$): In each round of the Algorand protocol, the sortition algorithm will select multiple (up to $70$) nodes as leaders, with each leader proposing its own block. Each committee member in each round, specifically, during the reduction phase of the Binary\BAStar~protocol, need to select (and vote) for the block with the highest priority. This block selection cost borne by a subset of committee members in each round, which includes the verification of sortition proofs, is parameterized as $c^{bs}$

%
    
 {\bf Vote Cost} ($c^{vo}$): Each selected committee member during each step of the Binary\BAStar~protocol should validate and check incoming messages (including, votes from previous steps) before submitting its own vote in that step. This cost, which also includes cost to append the sortition proof to the outgoing vote and broadcast to neighbors, is parameterized by us as $c^{vo}$.  The timeout to submit a vote is defined by Algorand and is equal to 20 seconds. It should be noted that $c^{vo}$ does not include the cost of vote counting and is outlined next.


    
{\bf Vote Counting Cost} ($c^{vc}$):  After all committee members have submitted their votes, each Algorand node should validate voting messages by checking their sortition proofs. $c^{vc}$ represents all associated costs of sortition proofs and signature verifications incurred when counting and tallying the votes inside each received vote message.



Given the above parameterization of some significant individual costs, let us outline the overall costs incurred by each individual Algorand node. Each node incurs a cumulation of two types of costs: \emph{(i)} a fixed cost, and \emph{(ii)} a role-based cost. The fixed cost ($c^{fix}$) represents the required costs borne by each node irrespective of its role and is given as: 
\begin{equation}
c^{fix}= c^{ve} + c^{se} + c^{so} + c^{go} + c^{vs} + c^{vc}.
\end{equation}

In addition to the fixed cost $c^{fix}$ in each round, each node incurs a cost based on its role(s) (i.e., Leader, Committee Member or None) in that round and is represented as follows:
\begin{equation}\label{eq_Costs}
c^j = \begin{cases}
c^{fix} + c^{bl} & j \in \mathbb{L} \\
c^{fix} + c^{bs} + c^{vo}& j \in \mathbb{M} \\
c^{fix} &  j \in \mathbb{K},  
\end{cases}
\end{equation}
%
%
where $\mathbb{L}$, $\mathbb{M}$, and $\mathbb{K}$ are the sets of leaders, committee members and all other users without particular rule, in round $i$.

\subsection{Reward Sharing in Algorand}
\label{sub:Reward}

Given the various costs, as outlined above, it is clear that rational users or participants (which we assume our users are) will fully participate in the distributed consensus protocol of Algorand if and only if they have enough incentive (or rewards) to do it. As Algorand is a cryptocurrency, the mechanism for providing incentives is straightforward - pay users in \emph{Algos} for their participation efforts and costs.  Bitcoin (and other cryptocurrencies) has an incentive model where the winner of the PoW puzzle receives incentives in the form of block rewards and transaction fees (paid out in Bitcoins) to be engaged in the PoW and block addition process. A similarly question arises in Algorand: \emph{which users should be paid, and how much, in order to enable their continued participation in the distributed consensus process?} 


Recently, the Algorand Foundation\footnote{The Algorand Foundation is dedicated to fulfilling the global promise of blockchain technology by leveraging the Algorand protocol and open source.} has suggested a tentative version of reward sharing in their protocol~\cite{algorand-token-dynamics},\cite{algorand-rewards-faq}, as shown in Figure~\ref{fig:RewardModel1}. The proposed reward sharing mechanism assumes creation and maintenance of two reward pools: \emph{(i) Foundation Reward Pool}, and \emph{(ii) Transaction Reward Pool}. These pools are nothing but public keys controlled by the Foundation. These public keys act as a central (foundation-controlled) storage where reward distribution related funds (\emph{Algos}) are deposited. All rewards for each round of the Algorand protocol are expected to be disbursed (or transferred) from this public key.
\begin{figure}[t]
\centerline{\includegraphics[scale=0.65]{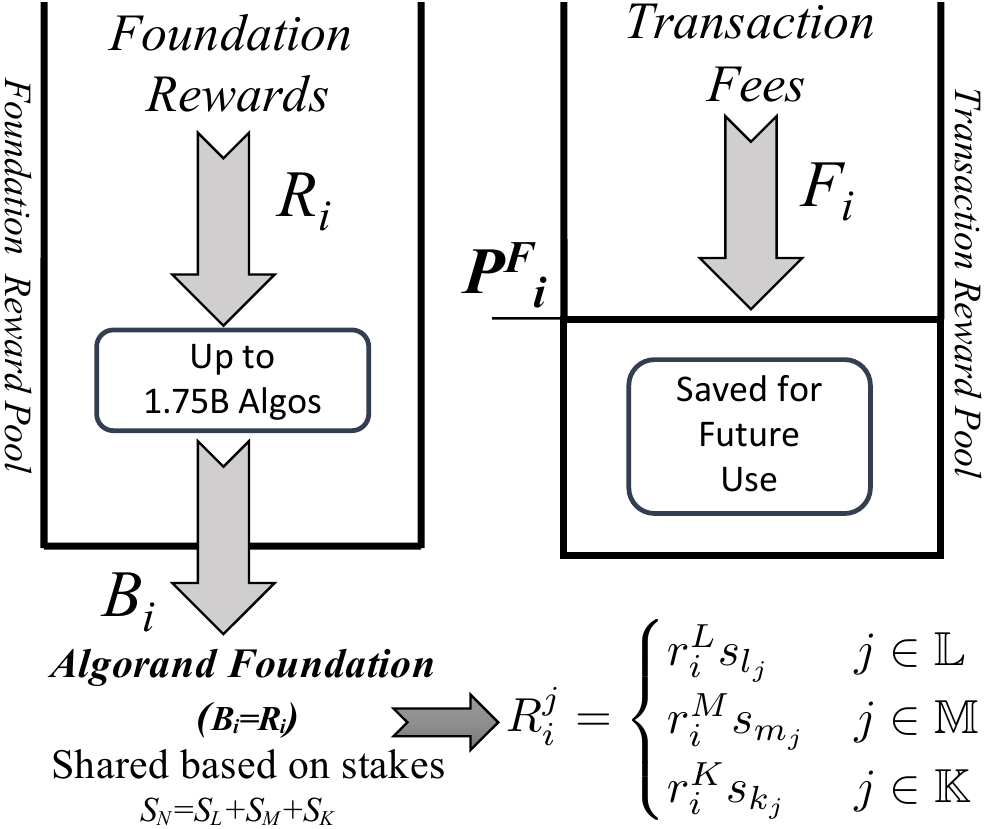}}
\caption{Algorand first shares the reward from 1.75B \emph{Algos}. Each time it shares $R_i$ among all users. We assume that the amount of reward taking out of the pool in round $i$ is equal to the input reward, i.e., $B_i=R_i$. Algorand saves all transaction fees in another pool to be used later.}
\label{fig:RewardModel1}
\end{figure}
To bootstrap the new cryptocurrency, the Algorand foundation implemented a ceiling of $1.75$ billion \emph{Algos} to be disbursed from the \emph{Foundation Reward Pool}. Per the foundation, in each round $R_i$ \emph{Algos} are added to the Foundation Reward Pool until the ceiling of $1.75$ billion is reached. According to Algorand Foundation, the projected rewards for the first $12$ reward periods follows the values presented in Table~\ref{tab:AlgorandReward}~\cite{algorand-token-dynamics}\cite{algorand-rewards-faq}. Each reward period spans $500$ thousands blocks. For example, in the first reward period, $10$ millions \emph{Algos} would be distributed, which is equal to approximately $20$ \emph{Algos} for each round, if in each round a block could be successfully added to the ledger.   

\begin{table*}[t]
	\centering
	\rowcolors{1}{gray!25}{white}
	\caption{Algorand Foundation suggested reward distribution in the first 12 reward period (equal to 6 millions blocks).}
		\begin{tabular}{l|c|c|c|c|c|c|c|c|c|c|c|c}
			\hline
			{\bf Reward Period} & 1&2&3&4&5&6&7&8&9&10&11&12 \\ \hline
			{\bf Projected Reward (Millions of \emph{Algos})} &10 &13&16&19&22& 25&28&31&34&36 &38&38\\ \hline
		\end{tabular}
	\label{tab:AlgorandReward}
\end{table*}

The reward sharing proposal suggests that in each round this reward $R_i$ is distributed among Algorand users in proportion to their current system stake, irrespective of their roles (i.e., leaders or committee member). In other words, users with higher stake receive a larger portion of the allocated foundation reward $R_i$ in each round. The transaction fees accumulated from the transactions in the added blocks during the bootstrap phase are saved or deposited into the \emph{Transaction Fee Pool}. This pool is not planned to be used for reward disbursement until the $1.75$ billion \emph{Algo} ceiling of the Foundation Reward pool is met. In summary, currently only the \emph{Foundation Reward Pool} is being used for the per-round reward (or incentive disbursement). Out of the $R_i$ \emph{Algos} disbursed to the \emph{Foundation Reward Pool} per round $i$, let us assume that $B_i$ (where, $B_i \leq R_i$) \emph{Algos} are actually disbursed among the participating or system users. Initially $B_i$ is expected to be equal to $R_i$.
Let us assume that the total value of stake in the system is $S_N$. Thus, $S_N= S_L + S_M + S_K.$ 
Here, $S_L$, $S_M$, and $S_K$ are the total stake values of leaders, committee members, and all other online nodes in round $i$, respectively. These values are changing in each round, but for the sake of presentation we write $S_N$ instead of $S_N(i)$. Then, the rewards assigned to a leader node $l_j$ in round $i$, $R_i^j$, would be $\frac{B_is_{l_j}}{S_N}$, where $s_{l_j}$ is the stake of leader $l_j$. In summary, we can define all reward distributions by:
\begin{equation}\label{eq_Rewards}
R_i^j = \begin{cases}
r_i^L s_{l_j} & j \in \mathbb{L} \\
r_i^M s_{m_j} & j \in \mathbb{M}\\
r_i^K s_{k_j} & j \in \mathbb{K}, 
\end{cases}
\end{equation}
where $ r_i^L = r_i^M = r_i^K = r_i = \frac{ B_i}{S_N} .$ Considering the proposed approach for reward sharing by the Algorand Foundation, we now analyze if the incentives provided by this mechanism is enough to guarantee cooperation by rational nodes.


\subsection{To be Cooperative or Not: Problem Motivation}
\label{sub:Motivation}

Let us assume that an Algorand node is \emph{cooperative} when it plays its role honestly by performing all of the assigned tasks, and consequently accepting all the associated costs. In contrast, a \emph{defective} node only remains online but does not perform any of its assigned tasks, except sortition computing to join the network (i.e., paying cost $c^{so}$). In this case, if appropriate countermeasures (e.g., punishment mechanisms) are not deployed, the defecting nodes may end up earning rewards by simply relying on other nodes to honestly perform their tasks and not contributing anything towards the block proposal, verification and consensus tasks. Considering this definition for cooperative and defective behavior, we can divide Algorand node behaviors into the following four categories: 

\begin{itemize}
\item {\bf Honest nodes:} These nodes always cooperate. They are also altruistic and cooperate even when the reward is not more than the cost of cooperation.   
\item {\bf Honest but Selfish nodes:} These nodes cooperate and defect depending on the amount of received incentives versus the cost for their actions. In other words, they are always selfish and will cooperate if and only if the reward is more than the cost of cooperation.
\item {\bf Malicious nodes:} They arbitrarily cooperate or defect. In addition to this, they may inject malicious transactions and blocks, or arbitrarily compromise other nodes.
\item {\bf Faulty nodes:} These nodes are offline due to system malfunction (and not by choice) and do not contribute anything to the network operation.
\end{itemize}

In this paper, we assume that all network nodes behave in an honest but selfish manner. Moreover, in this preliminary work, we assume that nodes do not arbitrarily behave maliciously or become faulty. In other words, nodes in our network make a strategic decision to cooperate (participate) or defect (not participate) solely by maximizing their own interests/incentives. They neither make any arbitrary protocol participation decision, nor maliciously modify the protocol to maximize their interests/incentives. To get an insight into the robustness of the proposed Algorand reward sharing approach against selfish (or rational) node/user behavior, we conduct some preliminary simulation experiments.



Our simulator, written in Python, is based on the Algorand discrete event simulator by Deka et al.~\cite{algosim} and implements all Algorand protocol modules, including, Sortition, Reduction and Binary\BAStar. We are also able to simulate network delays and various synchrony conditions, as well as, customize different network parameters such as total number of nodes and the distribution of network message delays in our simulator. Within this simulation framework, we also implemented the reward sharing protocol proposed by the Algorand Foundation (described earlier), which computes a per-round reward to be shared among the nodes. We simulate each round of the Algorand block proposal and consensus protocol, as outlined in Section~\ref{sec:SysModel} , and execute the reward sharing algorithm at the end of each round to compute and distribute rewards to all the network nodes based on the reward sharing protocol.


We simulate 100 runs of the protocol rounds, and average over all the possible block outcomes in 100 simulations. In each simulation instance, we randomly select defective nodes (i.e., honest and selfish nodes who chose to defect given their payoff) by means of a uniform distribution. We consider the total number of defective nodes in the network in steps of 5\%, 10\%,  15\%, 20\%, 25\%, and 30\% of all the nodes in Algorand network. Moreover, we distribute the stakes among all nodes with a uniform distribution between 1 to 50 \emph{Algos}. Note that we compute trimmed mean which ignores 20\% top and bottom data to compute the mean values of these 100 simulations. In our simulation each node sends the messages to 5 other nodes that are randomly selected from the network. We first analyze the impact of defective nodes on the block creation process.
 The corresponding number of nodes who extracted final, tentative, or no blocks from the network messages (i.e., votes) are plotted in Figure~\ref{fig:defection}. As shown in Figure~\ref{fig:defection}-(a), even with a low defection rate of 5\% the number of tentative blocks is increasing in the network. Moreover, about 7\% of nodes do not receive any block. When the number of defective nodes is increasing the Algorand network fails to inform most of nodes about the final blocks. 



\begin{figure*}[t]
	\centering
	\begin{subfigure}{.33\textwidth}
		\centering
		\includegraphics[width=\linewidth]{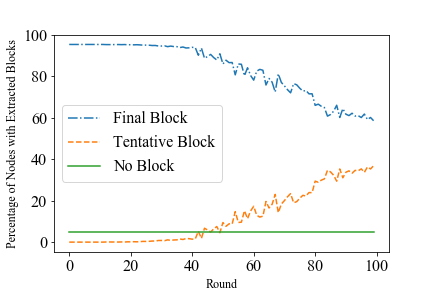}
		\caption{Defection Rate: $5\%$}
	\end{subfigure}
	\begin{subfigure}{.33\textwidth}
		\centering
		\includegraphics[width=\linewidth]{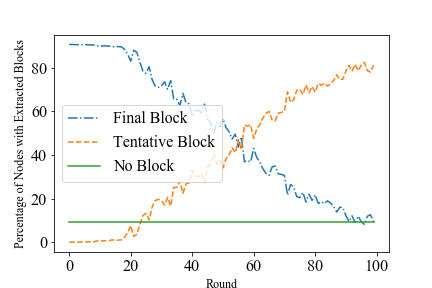}
		\caption{Defection Rate: $10\%$}
	\end{subfigure}%
	\begin{subfigure}{.33\textwidth}
		\centering
		\includegraphics[width=\linewidth]{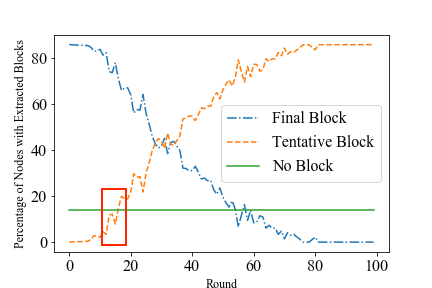}
		\caption{Defection Rate: $15\%$}
	\end{subfigure}\\
	\begin{subfigure}{.33\textwidth}
		\centering
		\includegraphics[width=\linewidth]{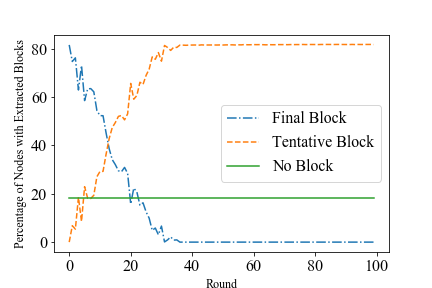}
		\caption{Defection Rate: $20\%$}
	\end{subfigure}
	\begin{subfigure}{.33\textwidth}
		\centering
		\includegraphics[width=\linewidth]{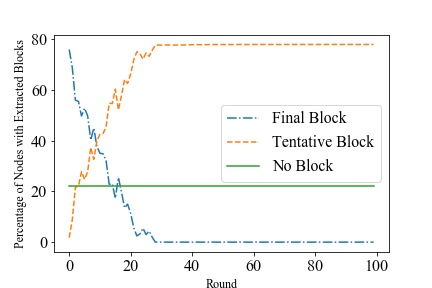}
		\caption{Defection Rate: $25\%$}
	\end{subfigure}%
	\begin{subfigure}{.33\textwidth}
		\centering
		\includegraphics[width=\linewidth]{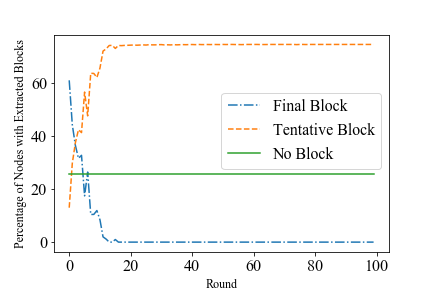}
		\caption{Defection Rate: $30\%$}
	\end{subfigure}%

	\caption{The percentage of nodes who extracted the tentative and final blocks with different rate of defection. In each scenario, we randomly choose nodes to behave as honest but selfish nodes. These nodes are called defective and will not cooperate if the benefit is not more than the cost of cooperation for them.}
	\label{fig:defection}
\end{figure*}


For example, as shown in Figure ~\ref{fig:defection}-(c), with $15\%$ defection rate, most of Algorand nodes don't reach any consensus on a final block after round $\#30$. 
In other words, with $15\%$ of defective nodes the network may go to weak synchrony state or even asynchrony in some rounds and it prevents some nodes to receive network messages (e.g., votes and block proposals). However, by reaching the strong synchronous network after a long period of asynchrony (i.e. weak synchrony assumption), nodes who have extracted tentative blocks can finalize their blocks. This effect has been highlighted in Figure~\ref{fig:defection}-(c) in the proximity of rounds 17 through 20. As shown in the figure, in round $\#17$ the asynchrony of network has caused an increase in the number of nodes that have extracted tentative blocks from the network. But in round $\#18$, network becomes synchronous again and consequently a majority of the Algorand nodes are able to extract the finalized blocks. We also need to clarify that these defective nodes may control more than threshold $h$ (i.e. Algorand honest assumption as defined in Section~\ref{subsub:Adversary}) of stakes in the network. This happens if there are more nodes with high values of stakes in the list of defective nodes. Defection of these nodes can amplify the network synchrony problem in the Algorand network and consequently the block creation process. Finally, the results show that with even 30\% of defective nodes the network fails even in the first few rounds.

%

In summary, the above simulation results show that without an incentive-compatible reward sharing approach that fosters cooperation, rational nodes will be inclined to defect from the block creation and consensus process resulting in asynchrony mode and fail to provide blocks. In the following section, we will propose a game-theoretical model to analyze the  effect of defection and propose a possible solution to avoid defective behavior in the Algorand network.

\section{Game Model and Incentive Analysis in Algorand} 
\label{sec:GameModel}

In order to obtain a good insight of the strategic behavior of the nodes in Algorand, we model the interaction of these nodes with a static non-cooperative game. We first focus on the interaction between the Algorand nodes who are supposed to interact and create blocks in each round. Let us assume that these nodes are the  players of a static game $\mathcal{G}^{Al}$ in round $i$ of the Algorand process. We assume that all strategies are hard-wired for each node. In other words, each node does not change his chosen strategy during round $i$ of the game. They also choose their strategies simultaneously. 

In our Algorand game $\mathcal{G}^{Al}$ users must decide whether to cooperate and contribute to make a new block or not. The game $\mathcal{G}^{Al}$ is defined as a triplet $(\mathbb{P},\mathbb{S},\mathbb{U})$, where $\mathbb{P}$ is the set of players, $\mathbb{S}$ is the set of strategies and $\mathbb{U}$ is the set of payoff values. The set of players $\mathbb{P}$ includes leaders $\mathbb{L}$, committee members $\mathbb{M}$, and all other users $\mathbb{K}$, i.e., $\mathbb{P} = \mathbb{L} \bigcup \mathbb{M} \bigcup \mathbb{K}$. An Algorand node can take an action ($s_i$) from the set $\mathbb{S} = \{C, D, O \}$, where $C$, $D$, and $O$ represent \emph{(i)} \emph{Cooperate}, \emph{(ii)} \emph{Defect}, and \emph{(iii)} \emph{Offline}, respectively. As we discussed in previous section, cooperative nodes follow all defined tasks, while defective nodes are only online but do not perform their assign tasks. Moreover, a node can play \emph{offline} in round $i$ (i.e., plays $O$), in which it runs sortition computing but it becomes offline and do not receive any reward. Given the above assumption, the following lemma shows that the $O$ strategy is always strictly dominated by $D$ strategy.

%

\begin{lemma}\label{lem:O}
In $\mathcal{G}^{Al}$, strategy $O$ is strictly dominated by playing defection ($D$).
\end{lemma}

\begin{proof}
A user always obtains greater payoffs by playing $D$ instead of $O$, for all possible strategy profiles of other users (i.e., opponents). In fact, a user can obtain the reward by playing $D$ in the current version of Algorand, but it's payoff would be $-c^{so}$, if it plays $O$.
\end{proof}

Given the result in Lemma~\ref{lem:O}, we are not going to consider strategy $O$ in our analysis as it cannot appear in any Nash equilibrium or will not be chosen by any rational player. In the following section, we present our results for the analysis of $\mathcal{G}^{Al}$. Recall that the definition of $\mathcal{G}^{Al}$ are based on the proposed reward sharing by Algorand Foundation~\cite{algorand-token-dynamics}.

\subsection{Analysis of $\mathcal{G}^{Al}$}

In $\mathcal{G}^{Al}$, we define strategies profiles $All-D$ and $All-C$, where all nodes choose to play $C$ and Play $D$, respectively. We apply the following game-theoretic concept to analyze $\mathcal{G}^{Al}$.  

\begin{definition} \label{def:NE}
The strategy profile $s^{*}$ constitutes a Nash equilibrium profile if none of the players
can unilaterally change his strategy to increase his utility. 
\end{definition} 

In other words, if all strategies are mutual best responses to each other, then no player has any motivation to deviate from the given strategy profile. The following theorem shows the existence of all defection strategy ($All-D$) NE for $\mathcal{G}^{Al}$.

\begin{theorem} \label{thm:AllD}
In each round $i$ of $\mathcal{G}^{Al}$ with $N$ players ($n_L$ leaders, $n_M$ committee members, and $n_K$ remaining nodes), all-defection strategy profile ($All-D$) is a Nash equilibrium.
\end{theorem}
\begin{proof}
Let us consider a strategy profile where all Algorand nodes defect, where there is no incurred costs such as  $c^L$, $c^M$, or $c^K$ for a leader, committee member, or other online nodes. Hence, the payoff for each node would be $u_i = -c^{so}$ as there is no block added to the chain and they cannot earn any \emph{Algo}. In this case:
\begin{enumerate}
    \item None of the Algorand leaders $l_j$ can increase their payoff unilaterally by changing their strategies. Because, the cooperative leader can not gain any reward without contribution of at least $S_{STEP}$ committee members in each step of \BAStar~protocol and $S_{FINAL}$ members for the final committee, as discussed in Section \ref{sec:SysModel}. In other words, the payoff of a leader who deviates from $D$ to $C$ would be $u_i^{l_j} (C) = -c^L$, which is always smaller than his defective payoff (i.e., $u_i^{l_j} (D) =-c^{so}$).
    
    \item Similarly, a cooperative committee member $m_j$ cannot obtain any reward without the contribution of leaders and sufficient number of committee members. In this case, payoff of a committee member who has deviated is $u_i^{m_j}=-c^M$.
    
    \item With similar justification, we can prove that all other online nodes $k_j$  will not be able to increase their payoffs unilaterally by deviating from $D$ to $C$, as its payoff would be decreased to $u_i^{k_j}(C) = -c^K$. 
\end{enumerate}
Hence, $All-D$ strategy profile is a NE in $\mathcal{G}^{Al}$.
\end{proof}
In fact, in such distributed protocols one would like to enforce $All-C$ strategy profiles as a Nash equilibrium. But the following theorem shows that the current Algorand incentive mechanism can not enforce all nodes to cooperate.
\begin{theorem}\label{thm:AllC}
In each round $i$ of $\mathcal{G}^{Al}$ with $N$ players ($n_L > 1$ leaders, $n_M$ committee members, and $n_K$ remaining nodes), if rewards are shared solely based on the current values of the stakes as shown in Equation (\ref{eq_Rewards}), i.e., the proposed Algorand Foundation mechanism, we cannot establish all-cooperation strategy profile (All-C) as a Nash equilibrium strategy profile.
\end{theorem}

\begin{proof}
First, let us assume that all Algorand nodes have already cooperated and paid the costs $c^L$, $c^M$, or $c^K$ as leader, committee member, or other online nodes. Given the values of costs and rewards calculated in Equation (\ref{eq_Costs})  and (\ref{eq_Rewards}), one can compute the payoff for each node $j \in \{l_j, m_j, k_j \}$ in round $i$ as follows:
\begin{equation}\label{eq_Payoffs}
    u_i^j(C) = \begin{cases}
        r_i s_{l_j} - c^L & \text{Leader } l_j \\
        r_i s_{m_j} - c^M & \text{Committee member } m_j \\
        r_i s_{k_j} - c^K & \text{Online node } k_j
    \end{cases}
\end{equation}
Consequently, by comparing cooperative and defective payoffs for each node, and if we assume that they deviate unilaterally, we can conclude that
\begin{enumerate}
    \item A leader $l_j$ can increase his payoff unilaterally by ignoring his role and acting as an online node without any role. In other words, it plays $D$ and acts as an online node. Because, other leaders are still active in the Algorand network and the protocol can reach consensus with remaining nodes. So, $l_j$ will pay the cost $c^{so}$ which is always smaller than $c^L$.
    \item Similarly, a committee member $m_j$ can increase his payoff unilaterally by ignoring his role and play $D$.
    \item Finally, a remaining member $k_j$ can increase his payoff by avoiding to pay $c^K$ costs and just keep himself online without gossiping messages (i.e., play $D$). Recall that it must still pay $c_{so}$ for sortition computing. In this case, the utility of node $k_j$ would be $u_i^{k_j} (D) = r_i s_{k_j} -c ^{so}$ which is greater than his previous payoff.
\end{enumerate}
Then, $All-C$ strategy profile can never be an NE in $\mathcal{G}^{Al}$.
\end{proof}

The results presented by Theorems~\ref{thm:AllD} and \ref{thm:AllC} show that we are cannot enforce cooperation in the current version of reward sharing in Algorand. In fact, if all users are rational they will try to only play $D$ and the system remains in $All-D$ Nash equilibrium. The main problem is due to have different costs for different type of nodes, whereas they receive the same portion of rewards. Hence, in the following we suggest a novel mechanism to share rewards in Algorand which considers different type of users.

 \subsection{Our Proposed Reward Sharing Mechanism}

As shown in Figure~\ref{fig:OurModel}, we suggest that the reward $B_i$ must be divided into three portions, and then be distributed among the nodes given their stakes. In our model, we assume that $\alpha B_i$, $\beta B_i$, and $\gamma B_i$ must be distributed among leaders, committee members, and other online nodes, where $\alpha \in (0,1)$, $\beta \in (0,1)$, and $\gamma \in (0,1)$ should be chosen by the designer, such that Algorand Foundation can enforce the cooperation among users. Note that $ \alpha + \beta + \gamma  =1$.  Given this approach, one can provide different incentives to different type of users. Hence, the payoff would be calculated by 
\begin{equation}\label{eq_PayoffsNew}
    u_i^j(C) = \begin{cases}
        r_i^L s_{l_j} - c^L & \text{Leader } l_j \\
        r_i^M s_{m_j} - c^M & \text{Committee member } m_j \\
        r_i^K s_{k_j} - c^K & \text{Online node } k_j,
    \end{cases}
\end{equation}
where $r_i^L = \frac{\alpha B_i}{S_L}$, $r_i^M = \frac{\beta B_i}{S_M}$, and $r_i^K = \frac{\gamma B_i}{S_K}$.  Let us now define and analyze a new game $\mathcal{G}^{Al+}$, in which the payoffs are calculated by Equation (\ref{eq_PayoffsNew}). 

\begin{figure}[t]
\centerline{\includegraphics[scale=0.55]{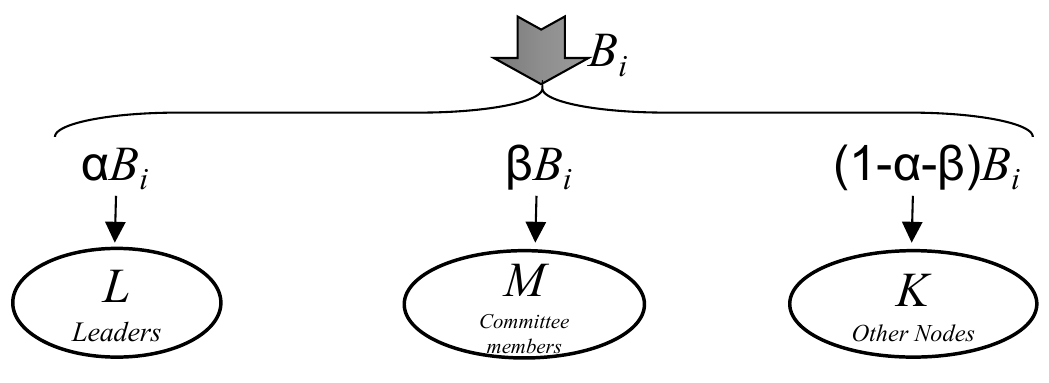}}
\caption{Our proposed model shares the reward according to the roles of nodes as well as their stakes.}
\label{fig:OurModel}
\end{figure}

\subsection{Analysis of $\mathcal{G}^{Al+}$}

In this section, we will first find conditions under which we can foster cooperation of users. Then we will investigate the existence of NE in this game. The following lemma presents the conditions under which, the leaders and the committee members have enough incentive to cooperate in round $i$.

\begin{lemma}\label{Lem:CooperationLeaderMember}
Considering $\mathcal{G}^{Al+}$ with N players ($n_L > 1$ leaders, $n_M$ committee members, and $n_K$ remaining nodes), where reward $B_i$ shares with ratios $\alpha$, $\beta$, and $\gamma = 1 - \alpha - \beta$ between leaders, committee members, and remaining nodes. A selfish leader $l_j$ or committee member $m_j$ cannot deviate from $C$ strategy unilaterally to increase its payoff, if and only if:
\begin{equation*}
    B_i > \max \{\frac{c^L - c^{so}}{(\frac{\alpha}{S_L} - \frac{\gamma}{S_K + s^*_{l_j}})s^*_{l_j}}, \frac{c^M- c^{so}}{(\frac{\beta}{S_M} - \frac{\gamma}{S_K + s^*_{m_j}})s^*_{m_j}}  \},
\end{equation*}
where $s^*_{l_j}$ and $s^*_{m_j}$ are the minimum values of stakes for the leaders and committee members in round $i$.
\end{lemma}
\begin{proof}
Let us consider that all leaders and committee members have cooperated in a given strategy profile. In this case, the payoff for any cooperative leader $l_j \in L$ would be equal to $u_i^{l_j}(C) = \frac{\alpha B_i}{S_L} s_{l_j} - c^L$. This payoff would be changed to $u_i^{l_j}(D) = \frac{\gamma B_i}{S_K + s_{l_j} } s_{l_j} - c^{so}$, if the leader $l_j$ plays $D$ and only keep its status online, without playing its role of a leader in Algorand. Hence, this leader has no incentive to defect if $u_i^{l_j}(C) > u_i^{l_j}(D)$. Consequently, we can show that under the following condition on $B_i$, the leader $l_j$ has no incentive to deviate from $C$ to $D$:
\begin{equation}\label{eq_BoundLeader}
B_i > \frac{c^L - c^{so}}{(\frac{\alpha}{S_L} - \frac{\gamma}{S_K + s_{l_j}})s_{l_j}}.
\end{equation}
%
%
%
Similarly, any committee member $m_j \in M$ cannot increase his payoff unilaterally by defecting and play $D$ if:
\begin{equation}\label{eq_BoundMember}
B_i > \frac{c^M - c^{so}}{(\frac{\beta}{S_M} - \frac{\gamma}{S_K + s_{m_j}})s_{m_j}}.
\end{equation}
%
Given two different bounds on the distributed rewards in Equations \eqref{eq_BoundLeader} and \eqref{eq_BoundMember}, and if we consider that $s^*_{l_j}$ and $s^*_{l_m}$ are the minimum values of stakes for leaders and committee members in round $i$, we can conclude that no leader or committee member can deviate in round $i$ if 
\begin{equation*}
    B_i > \max \{\frac{c^L - c^{so}}{(\frac{\alpha}{S_L} - \frac{\gamma}{S_K + s^*_{l_j}})s^*_{l_j}}, \frac{c^M- c^{so}}{(\frac{\beta}{S_M} - \frac{\gamma}{S_K + s^*_{m_j}})s^*_{m_j}}  \}.
\end{equation*}
\end{proof}
The results presented in Lemma~\ref{Lem:CooperationLeaderMember} show that the Algorand Foundation must always distribute enough rewards to the leader and the committee members in each round to enforce cooperative behaviors among them. The optimal reward $B_i$ is a function of the cost of cooperation and the current state of stakes in this round. It also depends on the values of $\alpha$, $\beta$, and $\gamma$, which must be selected by the administrator. We consider that these values would be announced at the beginning of each round. Another interesting fact is that if we assign more fraction of the reward $B_i$ to the leaders and the committee members (i.e., increasing the values of $\alpha$ and $\beta$), we can reduce the value of reward $B_i$, but have all leaders and committee members cooperative in a cooperative strategy profile. This will help the administrator to save more \emph{Algos} for future use. Finally, giving more rewards to online nodes ($k_j \in K$) will increase the value of required reward for cooperative behavior of leaders and committee members. Finally, it is worth mentioning that in Lemma~\ref{Lem:CooperationLeaderMember},  the following conditions must hold:
\begin{equation}\label{eq:Cond1}
\frac{\alpha}{S_L} - \frac{\gamma}{S_K + s_{l_j}} > 0
\end{equation}
\begin{equation}\label{eq:Cond2}
 \frac{\beta}{S_M} - \frac{\gamma}{S_K + s_{m_j}} > 0
\end{equation}
%
This can be easily proved given that the cost of cooperation for the leaders and the committee members (i.e., $c^L$ and $c^M$) are always positive. Having the required conditions on cooperative behavior of the leaders and the committee members in our hand with Lemma~\ref{Lem:CooperationLeaderMember}, we can now establish required conditions under which $\mathcal{G}^{Al+}$ has a Nash equilibrium apart from $All-D$ Nash equilibrium, in which some users cooperate in Algorand network. In fact this new class of cooperative NE in $\mathcal{G}^{Al+}$, will depend on the behavior of other online nodes and their characteristics in the Algorand network for any given round. Let us first review two important concepts defined by Algorand~\cite{gilad2017algorand} in the following definitions.  

%
%
%
%
%

\begin{definition}\label{Def_StrongSynch} 
In Algorand network, ``strong synchrony" is a network state, where most honest Algorand nodes (e.g., 95\%) can send messages that would be received by most of the other nodes (e.g., 95\%) within a given time limit.
\end{definition}
\begin{definition}\label{Def_WeakSynch} 
In Algorand network, with ``weak synchrony" state, the network can be asynchronous for a long but bounded period of time. After this asynchrony period, network must be again strongly synchronous for a reasonably long time.
\end{definition}
By having strong and weak synchrony definitions, we can form multiple sets of Algorand nodes which meet strongly synchronous network assumption together.
\begin{definition}
``Algorand strong synchrony set" is a list of Algorand nodes that together forms a strongly synchronous network.
\end{definition}
As Algorand protocol achieves liveness in strongly synchronous settings and safety with weak synchrony, the following theorem focuses on finding Nash equilibria for $\mathcal{G}^{Al+}$ with the strong synchrony assumption.
\begin{theorem}\label{thm:NASH1}
In game $\mathcal{G}^{Al+}$ with $N$ players ($n_L > 1$ leaders, $n_M$ committee members, and $n_K$ remaining nodes), for each Algorand strong synchrony set $\mathbb{Y}$, a strategy profile $s^*$  is a Nash equilibrium in round $i$, if in this strategy profile:
\begin{enumerate}
    \item All leaders cooperate,
    \item All committee members cooperate,
    \item All other nodes which are in $\mathbb{Y}$ cooperate, and
    \item Other online nodes defect
\end{enumerate}
and the value of $B_i$ is selected such that,
\begin{equation*}\begin{aligned}B_i > \max \{\frac{c^L-c^{so}}{(\frac{\alpha}{S_L} - \frac{\gamma}{S_K + s^*_{l_j}})s^*_{l_j}}, \frac{c^M-c^{so}}{(\frac{\beta}{S_M} - \frac{\gamma}{S_K + s^*_{m_j}})s^*_{m_j}}, \\ \frac{(c^K-c^{so}) S_K}{s^*_{k_j} \gamma}  \}\end{aligned}\end{equation*}
where $s^*_{l_j}$,  $s^*_{m_j}$, and $s^*_{k_j}$ are the minimum values of stakes for leaders, committee members, and other online nodes in $\mathbb{Y}$, in round $i$.
\end{theorem}

\begin{proof}
Let us assume that strategy profile $s^*$ is already played. We must now prove that none of the users can increase their payoffs unilaterally. The payoff of leaders and committee members cannot be increased unilaterally if the conditions of Lemma~\ref{Lem:CooperationLeaderMember} are hold. 
%
For each remaining node $s_{k_j}$ which is the member of Algorand strong synchrony set $\mathbb{Y}$, the payoff of cooperation would be $u_i^{k_j}(C) = \frac{\gamma B_i}{S_K} s_{k_j} - c^K$. But, the payoff of a defective node would be $u_i^{k_j}(D) = -c^{so}$. In other words, if a member of $\mathbb{Y}$ defects, no new final block would be created in this round given Definition~\ref{Def_StrongSynch}. So, to prevent $s_{k_j}$ from defecting:
\begin{align}
             & u_i^{k_j}(C) > u_i^{k_j}(D) \nonumber \\
    \implies & \frac{\gamma B_i}{S_K} s_{k_j} - c^K > -c^{so} \nonumber \\
    \implies & B_i > \frac{(c^K-c^{so}) S_K}{s_{k_j} \gamma} \label{thm5:eq3}
\end{align}
Hence this would be added to the conditions of Lemma~\ref{Lem:CooperationLeaderMember} to form the Nash equilibrium strategy profile of the game. Finally, other online nodes who are not in Algorand strong synchrony set can not increase their payoffs by deviating from $D$ to $C$ and accept the incurred cost of $c^K$, as the block would be made whether they cooperate or not. 
\end{proof}

\subsection{Proposed Reward Sharing Mechanism}

Our next goal is to extend the current Algorand reward sharing method by considering strategic behavior of users/nodes. In this case, we provide a solution for Algorand Foundation to foster cooperative behavior among all Algorand nodes, by sharing rewards based on the assigned roles. Moreover, our computed bounds in Theorem \ref{thm:NASH1} shows that we can minimize the reward $B_i$ by selecting suitable values for $\alpha$, $\beta$, and $\gamma$.

Our results presented in Section~\ref{sub:Motivation} showed that the Algorand Foundation needs to deploy an incentive-compatible mechanism to prevent nodes from selfish behavior and defect to increase its payoff unilaterally. We have proposed an algorithm based on Theorem~\ref{thm:NASH1} which provides enough incentive for Algorand nodes to cooperate. Our proposed Algorithm~\ref{reward-sharing-algorithm} proceeds as follows: at the end of each Algorand round, the Algorand Foundation computes the stakes for the leaders, the committee members, and other online nodes as  $S_L$, $S_M$, and $S_K$. It also computes the minimum stakes for each role as $s^*_{l_j}$, $s^*_{m_j}$, and $s^*_{k_j}$. Finally, foundation will calculate the optimal values for $\alpha$ and $\beta$ to minimize $B_i$, by using the defined bounds in Theorem~\ref{thm:NASH1}. As these values will be computed at the end of each round, all Algorand nodes know in advance that they cannot deviate from cooperation to obtain better payoff. Hence, the mechanism is strategy-proofed and there is no incentive for nodes to defect. 

\begin{algorithm}[t]
\caption{Incentive-Compatible Reward Sharing}\label{reward-sharing-algorithm}
\begin{algorithmic}[1]
   \Procedure{ComputeParameters}{$L$, $M$, $K$, $Stakes$}
        \State // Compute Stakes \& Costs
         \State $S_L\gets \sum_{l \in L}{s_l}$
        \State $S_M\gets \sum_{m \in M}{s_m}$
        \State $S_K\gets \sum_{k \in K}{s_k}$
        \State $c^L, c^M, c^K \gets$ \Call{AlgorandCosts}{ }
        \State // Compute minimum $s_{l_j}$ and $s_{m_j}$
        \State $s^*_{l_j} \gets \operatorname{min}_{l_j \in L}{s_{l_j}}$
        \State $s^*_{m_j} \gets \operatorname{min}_{m_j \in M}{s_{m_j}}$
        \State $s^*_{k_j} \gets \operatorname{min}_{k_j \in K}{s_{k_j}}$
        \State // Compute $\alpha$, $\beta$, $B_i$ from Theorem \ref{thm:NASH1} bounds
        \State Find $\alpha$, $\beta$ to minimize $B_i$ where: 
        \State \vspace{-0.53cm}\quad \begin{small}\begin{equation*}\begin{aligned}B_i > \max \{\frac{c^L-c^{so}}{(\frac{\alpha}{S_L} - \frac{\gamma}{S_K + s^*_{l_j}})s^*_{l_j}}, \frac{c^M-c^{so}}{(\frac{\beta}{S_M} - \frac{\gamma}{S_K + s^*_{m_j}})s^*_{m_j}}, \\ \frac{(c^K-c^{so}) S_K}{s^*_{k_j} \gamma}  \}\end{aligned}\vspace{-0.51cm}\end{equation*}\end{small} 
        \Return $\alpha$, $\beta$, $B_i$
    \EndProcedure
\end{algorithmic}
\end{algorithm}

\section{Evaluation}
\label{sec:Simulation}

In order to evaluate our proposed mechanism, we first conduct a series of numerical analysis to obtain the best reward shares in our model (i.e., $\alpha$ and $\beta$). We then deploy our proposed reward sharing mechanism over Algorand and evaluate its performance, comparing to the reward sharing proposed by Algorand Foundation. 
  
\subsection{Optimal Reward Share Calculation: A Numerical Analysis}

According to the results presented in Theorem~\ref{thm:NASH1}, we can minimize the reward in each round such that it guarantees the cooperation of subset of Algorand nodes. The optimal reward is ensured by choosing optimal reward shares for leaders and committee members, i.e., $\alpha$ and $\beta$. In our numerical analysis, we assume that the minimum acceptable values of stakes for each role are equal to $s^*_l=1$, $s^*_m=1$, and $s^*_k=10$ \emph{Algos}. In other words, by setting $s^*_k = 10$, we ignore any strong synchrony set containing nodes with stakes less than $10$ \emph{Algos} in our simulations. We also assume that the cost of cooperation for the leaders, the committee members, and other nodes are $c^L=16$, $c^M=12$, $c^K=6$, and $c^{so} = 5$ micro \emph{Algos}. Figure~\ref{fig:alpha-beta-b} shows the results for the minimum values of $B_i$ in each round,  for different values of $\alpha$ and $\beta$. Our results show that for $(\alpha, \beta)=(0.02, 0.03)$, the minimum values of $B_i$ would be about $5.2$ \emph{Algos} per round.

\begin{figure}[t]
	\centering
		\includegraphics[scale=0.28]{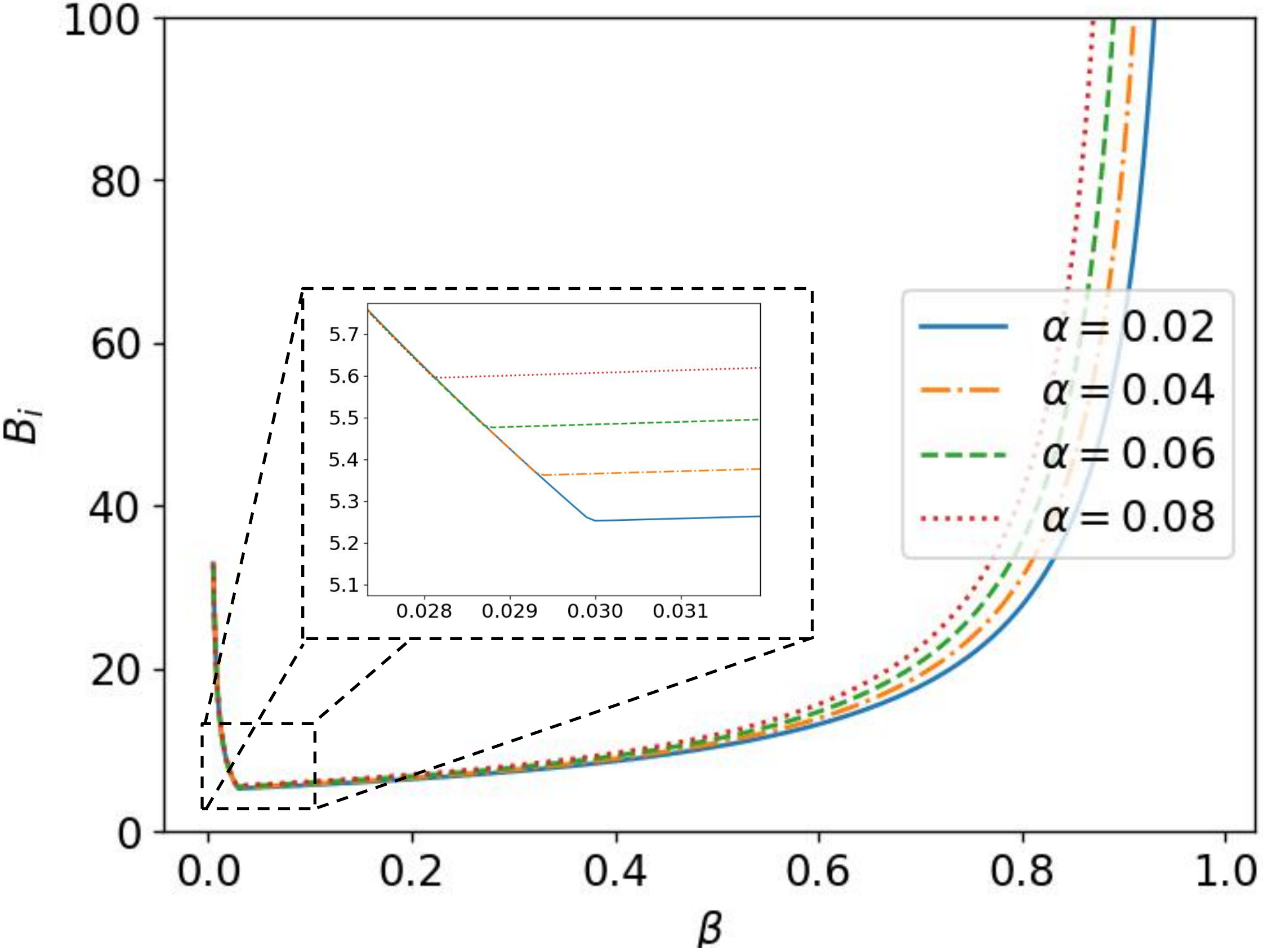}
	\caption{Minimum values of $B_i$ for different $\alpha$ and $\beta$ values.}
	\label{fig:alpha-beta-b}
\end{figure}

In fact, considering the value of $S_K$ which is always much greater than $S_L$ and $S_M$, the calculated bounds presented in Theorem~\ref{thm:NASH1} is usually a function of third bound, i.e., $\frac{(c^K -c^{so}) S_K}{s^*_{k_j} \gamma}$. Hence, to minimize the value of $B_i$ we need to maximize gamma and consequently minimize $\alpha$ and $\beta$ (recall that $\gamma = 1-\alpha - \beta$). In summary, we should always consider enough share of the total reward for leaders and committee members, as shown in Equations (\ref{eq:Cond1}) and (\ref{eq:Cond2}). Moreover, we also provide enough rewards to all other online nodes considering the value $B_i$ which is greater than $\frac{(c^K-c^{so}) S_K}{s^*_{k_j} \gamma}$.



\begin{figure*}
	\centering
	\begin{subfigure}{.24\linewidth}
		\centering
		\includegraphics[width=\linewidth]{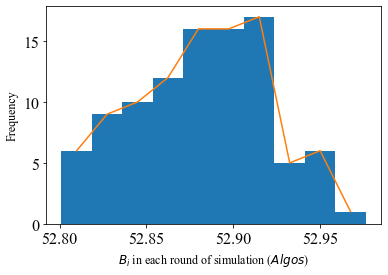}
		\caption{$\mathcal{U}(1,\,200)$}
	\end{subfigure}%
	\begin{subfigure}{.24\linewidth}
		\centering
		\includegraphics[width=\linewidth]{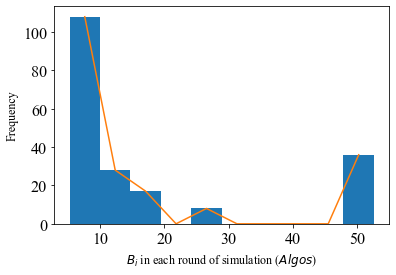}
		\caption{$\mathcal{N}(100,\,20)$}
	\end{subfigure}%
	\begin{subfigure}{.24\linewidth}
		\centering
		\includegraphics[width=\linewidth]{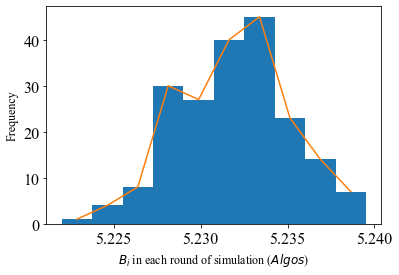}
		\caption{$\mathcal{N}(100,\,10)$}
	\end{subfigure}%
	\begin{subfigure}{.24\linewidth}
		\centering
		\includegraphics[width=\linewidth]{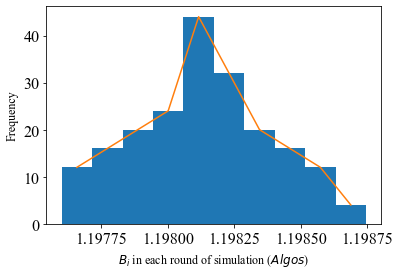}
		\caption{$\mathcal{N}(2000,\,25)$}
	\end{subfigure}%
	
	\caption{Distribution of computed $B_i$ values in each simulation by our proposed mechanism, for different distributions of stakes. }
	\label{fig:DistributionofStakes}
\end{figure*}

\begin{figure*}[t]
	\centering
	\begin{subfigure}{.33\linewidth}
		\centering
		\includegraphics[width=\linewidth]{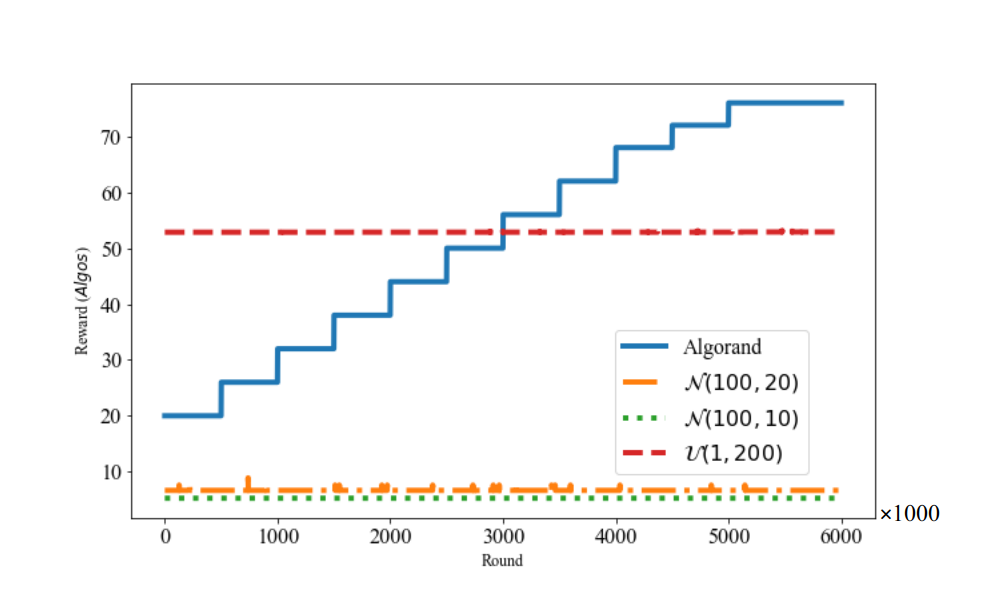}
		\caption{}
	\end{subfigure}%
	\begin{subfigure}{.33\linewidth}
		\centering
		\includegraphics[width=\linewidth]{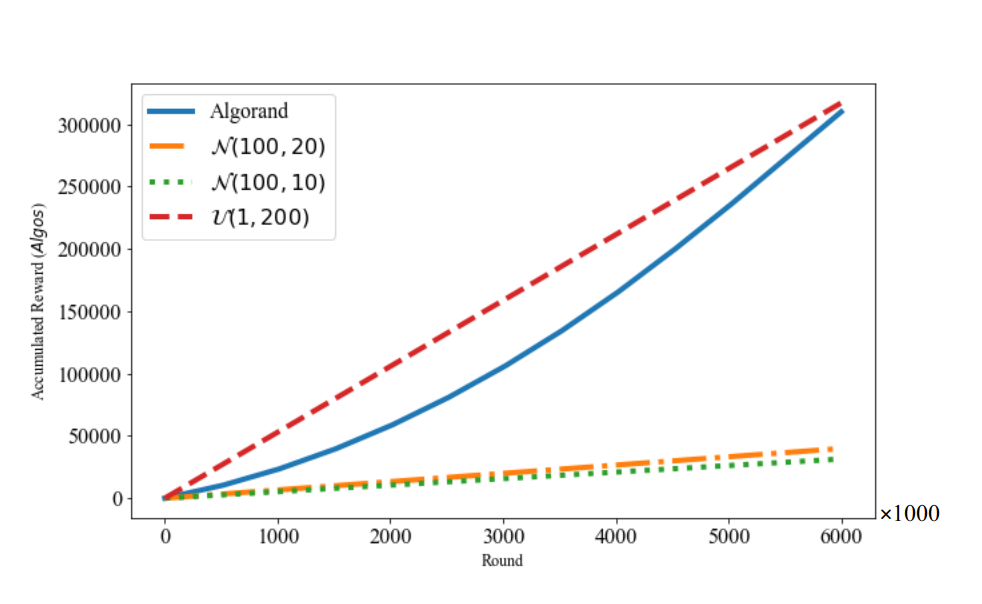}
		\caption{}
	\end{subfigure}%
	\begin{subfigure}{.33\linewidth}
		\centering
		\includegraphics[width=\linewidth]{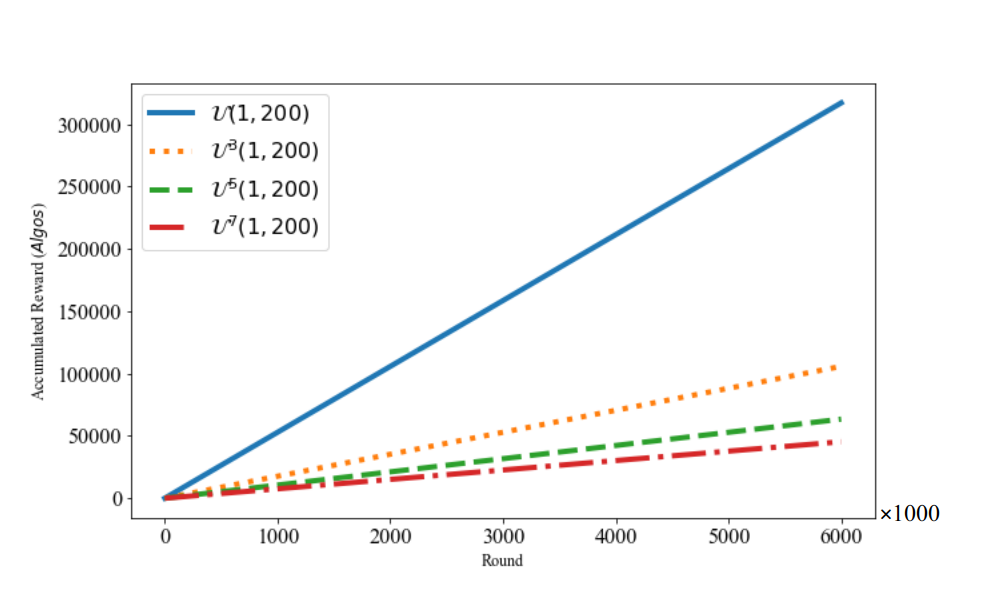}
		\caption{}
	\end{subfigure}%
	\caption{(a) Distributed rewards in each round by our proposed algorithm and Algorand Foundation given different distribution of stakes, (b) the accumulated rewards distributed among Algorand nodes, (c) the accumulated rewards  when the algorand nodes with less than 3 ($\mathcal{U}^3(1,\,200)$), 5 ($\mathcal{U}^5(1,\,200)$), and 7 ($\mathcal{U}^7(1,\,200)$) stakes have been removed from the network.}
	\label{fig:Reward}
\end{figure*}

\subsection{Performance Evaluation of Our Proposed Mechanism }

Given the calculated values for $\alpha$ and $\beta$, we evaluate the performance of our proposed mechanism in term of reward sharing and fostering cooperation. We simulate an Algorand network containing $500$ thousands nodes, in which the amount of stakes for leaders and committee members are $S_L = 26$ and $S_M = 13$K, respectively. In fact, $S_M$ is equal to  $S_{STEP}  \times (2 + 1) + S_{FINAL} \times 1$, according to Algorand,
where $S_{STEP}=1$K, and $S_{FINAL}=10$K~\cite{gilad2017algorand}. In other words, $S_{STEP}$ and $S_{FINAL}$ define bounds on the expected values of \emph{Algo} in a given round for reduction and Binary\BAStar. In our simulation, we distribute $50$ millions \emph{Algos} among these $500$K nodes using uniform  distribution of $\mathcal{U}(1,\,200)$ and normal distribution of $\mathcal{N}(100,\,20)$, $\mathcal{N}(100,\,10)$. This is similar to the initial phase of Algorand as well as the current status of the Algorand network. We also simulated transactions between all nodes. In each round, we choose randomly $1000$ nodes, in which nodes with higher stakes would be selected more often. Note that a node can be chosen more than one time in each round. Then we generate a series of random transactions for selected nodes with a uniform distributions between $-4$ to $4$. Negative values  represent sending \emph{Algos} while positive values represent receiving \emph{Algos} in nodes. With these values we tried to emulate the real Algorand exchange system available at algoexplorer~\cite{algo-explorer}. As for reward sharing mechanism, we deploy Algorand Foundation proposal presented in Table~\ref{fig:RewardModel1}  and our proposed mechanism presented in Algorithm~\ref{reward-sharing-algorithm}. We run the simulation for each graph for 200 times with different distributions, where each instance executes for 10 rounds. We finally made an average of rewards in each round.


The results show that the calculated rewards by our proposed mechanism follows the distribution of stakes in the network, as shown in Figure~\ref{fig:DistributionofStakes}. For example, we  must distribute high reward (around 50 \emph{Algos}) for uniform distribution of $\mathcal{U}(1,\,200)$, as there exist many nodes with low stakes. But with a normal distribution $\mathcal{N}(100,\,10)$, we need only to distribute small rewards, i.e., around 5 \emph{Algos}. In fact $\mathcal{N}(100,\,10)$ simulates the initial phase of Algorand, where around 50 millions \emph{Algos} were in the network. Comparing the results presented in Figure~\ref{fig:DistributionofStakes}-(d) (which simulates current status of Algorand with more than 1 billion \emph{Algos}, by $\mathcal{N}(2000,\,25)$ stake distribution) with Figures~\ref{fig:DistributionofStakes}-(c), we can also conclude that when the number of \emph{Algo} increases, we need smaller reward (around 1.2 \emph{Algos}) to enforce cooperation. The results show that the Foundation can  adapt the rewards given the status of the network in term of stakes, by using our model.


Figure~\ref{fig:Reward}(a)-(b) show the exact calculated reward in each round with our proposed algorithm and Algorand Foundation. As it is discussed in Section~\ref{sub:Reward}, Algorand Foundation suggests simple increasing reward distribution, however our approach adaptively choose the minimum possible reward to guarantee cooperation. The results show that our proposed mechanism distributes much smaller rewards among nodes, given the distribution of stakes. For example, in contrast to the proposed reward sharing by Algorand Foundation which shares $20$ \emph{Algos} in each round for the first 500 thousands rounds\cite{algorand-token-dynamics}, our proposed reward sharing algorithm will share about $5.2$ \emph{Algos} for normal distributions of stakes. More interestingly, our proposal will not increase the reward till 6 millions blocks generation, as it can guarantee cooperation without paying more \emph{Algos}. Our approach only distributes more rewards when the distribution of stakes is $\mathcal{U}(1,\,200)$ (Figures \ref{fig:DistributionofStakes}-(a)~and~\ref{fig:Reward}(a)-(b)). This is due to the fact that the number of nodes with small values of stakes are higher with this distribution.  If we remove the nodes with smaller stakes, e.g., up to 7 stakes from the set of rewarded nodes we can still keep the synchrony of network and distribute much smaller reward. This is shown in Fig.~\ref{fig:Reward}-(c), (where we remove nodes with stakes up to $w$, i.e., $\mathcal{U}^w(1,\,200)$). This also is useful for the Foundation, to tune the reward sharing according to the current conditions of stakes network.

\section{Conclusion}
\label{sec:Conclusion}

In this paper, we first introduced a system model to capture the main operational features of Algorand blockchain. We have  then presented a game-theoretical model to analyze the impact of defective behavior in Algorand blockchain. We then defined different behaviors of nodes in Algorand and show that there are no mutual cooperation states in the game where an Algorand user cannot increase its payoff by unilaterally defecting and not participating in the protocol tasks. We then comprehensively studied the problem of selfishness and proposed possible solution for it. We analyzed the defined game and obtained the Nash equilibria for different conditions. Our analytical results shows that we can always enforce cooperation by choosing enough rewards among Algorand nodes. Moreover, our numerical analysis validated that the proposed reward sharing mechanism outperforms the current proposal by Algorand Foundation. We believe that this work is the first step towards a better understanding of the effects of selfish behavior in Algorand network. As the results show, our mechanism can help the Algorand Foundation use the \emph{Algos} wisely as well as adapt dynamically with the distribution of stakes in the network. 

In term of future work, we can also get in touch with the Algorand Foundation to introduce our proposed mechanism for reward sharing in the initial phase (for 1.75 billion \emph{Algos}), as well as the distribution of transaction fees as reward in near future.

\balance

\bibliographystyle{ieeetr}
\bibliography{ref}

\end{document}